\documentclass{article}
\usepackage{amsmath}
\usepackage{amsfonts}
\usepackage{amssymb}
\usepackage{pstricks,pst-node,amssymb,amsmath,graphics,latexsym,tabularx,shapepar}
\usepackage{graphicx}
\usepackage{float}
\usepackage{float}
\usepackage{tikz-cd}
\usepackage{amscd}
\newcommand{\SN}{\mathrm{SN}}
\newcommand{\II}{\Pi}
\newcommand{\NN}{\mathbf{N}}
\newcommand{\LL}{\Lambda}
\newtheorem{theorem}{Theorem}

\newtheorem{corollary}[theorem]{Corollary}

\newtheorem{definition}[theorem]{Definition}
\newtheorem{example}[theorem]{Example}

\newtheorem{lemma}[theorem]{Lemma}

\newtheorem{proposition}[theorem]{Proposition}
\newtheorem{remark}[theorem]{Remark}

\newenvironment{proof}[1][Proof]{\noindent\textbf{#1.} }{\ \rule{0.5em}{0.5em}}
\newcommand{\bpartial}{\mathop{\partial\kern -4pt\raisebox{.8pt}{$|$}}}
\newcommand{\bra}{\mathopen{[\kern-1.6pt[}}
\newcommand{\ket}{\mathclose{]\kern-1.5pt]}}
\newcommand{\bbra}{\mathopen{[\kern-2.2pt[\kern-2.3pt[}}
\newcommand{\bket}{\mathclose{]\kern-2.1pt]\kern-2.3pt]}}

\makeindex

\begin{document}
\title {\large{\bf Right-invariant Poisson–Nijenhuis structures on Lie groupoids:\\ Correspondence and Classification}}
\vspace{2mm}
\author {\small{ \bf  Ghorbanali Haghighatdoost}\hspace{-1mm}\\
{\small{\em Department of Mathematics,Azarbaijan Shahid Madani University, Tabriz, Iran}}\\
           }
 \maketitle
 
\begin{abstract}
In this paper, we introduce right-invariant Poisson-Nijenhuis  Structures on Lie groupoids and their infinitesimal counterparts as called $(\wedge , \mathbf{n})-$structures. Also, we present a one-to-one correspondence between $(\LL  ,\mathbf{n})-$structures on Lie algebroids with Poisson-Nijenhuis structures $(\Pi , \mathbf{N})$ on their Lie groupoids under certian conditions. Also, we give some illustrative examples .
\end{abstract}

\section{Introduction}

As it is mentioned in the prologue of \cite{Kiril0Mackenzie}, groupoids possess many of the features which give groups their power and importance, and the modern concept of Lie groupoid is as much entitled as is the familiar concept of Lie group to be regarded as the rigorous formulation of the $19^{\text{th}}$ century notion which went under the vague term ``continuous group of local transformations''; a case could be made that the modern concept of Lie group has been a transitional stage in the evolution of the notion of Lie groupoid.

Recently, many researchers working on geometric structures on Lie groupoids and try to extend known methods on Lie groups to Lie groupoids. Poisson–Nijenhuis (P-N) structures on manifolds were introduced by Magri and Morosi \cite{Magri} (see also \cite{Kosmann3}), as Poisson and Nijenhuis structures which are compatible in a sense.

In \cite{Ravanpak, Haghighat} we studied right-invariant Poisson–Nijenhuis (Poisson–quasi-Nijenhuis) structures on a Lie group $G$ and introduced their infinitesimal counterpart, the so-called $r$-$n$ ($r$-$qn$) structures on the corresponding Lie algebra $\mathfrak{g}$.

In line with our previous works \cite{Ravanpak, Haghighat}, here we study Poisson–Nijenhuis structures on a Lie groupoid $G \rightrightarrows M$ and define right-invariant Poisson–Nijenhuis structures on it, calling them right-invariant PN-structures on the Lie groupoid $G \rightrightarrows M$. We also define the infinitesimal counterpart of this, called algebraic structures corresponding to Poisson–Nijenhuis Lie groupoids.

We emphasize that this work is concerned with continuous Poisson–Nijenhuis structures on Lie groupoids, which is distinct from discrete Lagrangian or Hamiltonian mechanics on groupoids..

In section 2 we review some preliminary definitions and theorems about Poisson–Nijenhuis manifolds and groupoids. Moreover, we recall basic concepts about Lie groupoids and Lie algebroids (for more details refer to \cite{Magri, Kosmann, Kosman2, Kiril0Mackenzie}). In subsections 2.1.1 and 2.1.2 we verify the tangent and the cotangent bundle of a Lie groupoid as new Lie groupoids and introduce their Lie groupoid structures.

In section 3 we define right-invariant Poisson–Nijenhuis structures on Lie groupoids and prove that right-invariant $(\Pi, \mathbf{N})$-structures on Lie groupoids are in one-to-one correspondence with $(\Lambda, n)$-structures on their Lie algebroids.

By some illustrative examples we end this work in section 4.
\section{ Main Concepts }
\subsection{ Lie groupoid and Lie algebroid }

As it is well-known,  a  groupoid  which is  denoted by $G \rightrightarrows M$, consists of two sets $G$ and $M$ together with structural mappings $\alpha, \beta, 1, \iota$ and $m,$  where source mapping $\alpha: G \rightarrow M,$ target mapping $\beta:G \longrightarrow M,$ unit mapping $1: M \longrightarrow G,$ inverse mapping $\iota : G \longrightarrow G $ and multiplication mapping $ m: G_{2} \longrightarrow G $ which $ G_{2} = \lbrace (g,h) \in G \times G ~\vert ~ \alpha (g) = \beta (h) \rbrace $ is a subset of $G \times G.$\\

A Lie groupoid is a groupoid $G \rightrightarrows M$ for which $G$ and $M$ are smooth manifolds, $\alpha, \beta, 1, \iota$ and $m$ are differentiable mappings and besides, $\alpha,\beta$ are differentiable submersions.

The right translations $R_g: G_{\beta(g)}=\alpha^{-1}(\beta(g))\to G_{\alpha(g)}=\alpha^{-1}(\alpha(g))$ on a Lie groupoid $G$ over $M$, are diffeomorphisms of the $\alpha$-fibers only and not of the whole groupoid. 

A smooth mapping $X: G\to TG$ is called a vector field on $G$, i.e. for every $g\in G$, $X(g)\in TG$. Where $TG$ is tangent bundle of $G$.

According to the above, to talk about right- invariant vector fields on $G$, we have to restrict attention to those vector fields which are tangent to the $\alpha$-fibers. 
In other words, we take the elements of the sections of the sub-bundle $T^\alpha G$ of $TG$ defined $$T^\alpha G=\mathrm{Ker}(d\alpha)\subset TG.$$

A Lie algebroid A over a manifold M is a vector bundle  $\tau : A \longrightarrow M$ with the following items:
\begin{enumerate}
\item A Lie bracket $[\vert ~,~ \vert] $ on the space of smooth sections $ \Gamma( \tau),$
$$[\vert ~,~ \vert]: \Gamma (\tau) \times \Gamma (\tau) \longrightarrow \Gamma (\tau), \quad (X , Y ) \longmapsto [\vert X , Y \vert].$$
\item A morphism of vector bundles $\rho : A \longrightarrow TM,$ called the anchor map, where $TM$ is the tangent bundle of $M,$
such that the anchor and the bracket satisfy the following Leibniz rule:
$$[\vert X , f Y \vert] = f [\vert X , Y \vert] + \rho (X) (f) Y,$$
where $X , Y \in \Gamma (\tau)$, $~f \in C^{\infty} (M)$ and $\rho (X) f$ is the derivative of $f$  along the vector field $\rho (X).$
\end{enumerate}
Given a Lie groupoid $G$ over $M$, we define the vector bundle $A=Lie(G)=AG$, whose fiber at $x\in M$ coincide with the tangent space at the unit $1_x$ of the $\alpha$- fiber at $x$. In other words, $AG:=T^\alpha G|_M$.

It is easy to see that every fiber of the sub-bundle $T^\alpha G$  at an arrow  $h: y\to z$ of $G$ is $T^{\alpha}_h G= T_hG(y,-)$,  where  $G(y,-)=\alpha^{-1}(y) = G_{\alpha(h)}$. Consider the right translation $R_g:G(y, -) \to G(x, -), g' \to g'g$.
The differential of the right translation by $g$ induces a map $ dR_g: T^\alpha_h G \to T^\alpha_{hg} G$.
\begin{definition}
Vector field $X$ on $G$ is called vertical if it is vertical with respect to $\alpha$, that is, $X_g\in T_gG_{\alpha(g)}$, for all $g\in G$. We call  $X$  right- invariant on $G$ if it is vertical and $X_{gh}=dR_g(X_h)$, for all $(h,g)\in G_{(2)}$. 
\end{definition}
\begin{remark}
It is easy to show that  $\Gamma(AG)$ - the space of sections of vector bundle $AG$ can be identified the space of right- invariant vector field on $G$.  So if we denote  the space of right- invariant vector field on $G$ by 
$$\chi^\alpha_{\mathrm{inv}}(G)=\{X\in \Gamma(T^\alpha G): \overrightarrow{X}_{hg}=dR_g(X_h), \;(h, g)\in G_{(2)}\}.$$
From above we have the space of sections  $\Gamma(AG)$ is isomorphic to the space of right- invariant vector fields on $G$, $\chi^\alpha_{\mathrm{inv}}(G)$. On the other hand, the space $\chi^\alpha_{\mathrm{inv}}(G)$ is a Lie sub-algebra of the Lie algebra $\chi(G)$ of vector field on $G$ with respect to the usual Lie bracket of vector fields. Also the pull-back of vector field on the $\alpha-$fibers along $R_g$, preserves brackets. So we obtain a new bracket on $\Gamma(AG)$ which is uniquely determined. The Lie bracket on $AG$ is the Lie bracket on  $\Gamma(AG)$ obtained from the Lie bracket on $\chi^\alpha_{\mathrm{inv}}(G)$. The anchor of $AG$ is the differential of the target mapping $\beta$, i.e. $\rho=T\beta|_{AG}: AG \to TM$. As a result, we obtain that $AG$ is a Lie algebroid associated to the Lie groupoid $G$.
\label{SVB}
\end{remark}
According to the remark  \ref{SVB} we have following lemma:
\begin{lemma}\label{lem0}
Right- invariant vector fields on the Lie groupoid $G \rightrightarrows M$ are identified with the sections of Lie algebroid $AG$.
\end{lemma}

\subsubsection{ Tangent Lie groupoid}

Let $G\rightrightarrows M$ be a Lie groupoid with the structural mappings  $\alpha, \beta, 1, \iota$ and $m,$ . Appling the tangent functor to each of the maps defining $G$ yields a Lie groupoid structure on $TG$  with base $TM$, source $T\alpha$ and target $T\beta$, and multiplication $Tm: TG_2\to TG$. The identity at $X \in TM$ is $T1(X)$. This is the tangent prolongation groupoid of $G \rightrightarrows M$, usually is omitted the word prolongation.

The multiplication $Tm$ is defined as follows:
$$ AB=Tm(A, B)=TL_\delta(B)+ TR_\tau(A)-TL_\delta TR_\tau(T1(w)),$$
where $\delta$ and $\tau$ are any (local) bisections of $G$ for which $\delta(\alpha (g))=g$ and $\tau(\delta h)=h$  for all $g, h \in G.$         
$A\in T_h(G), B\in T_g(G)$ and $T\alpha(A)=T\beta(B)=w.$
It is worth to say, a bisection of the Lie groupoid $G\rightrightarrows M$ is a smooth map $\delta: M \to G$ which is right-inverse to $\alpha: G \to M$ and is such that $\beta \circ \delta: M \to M$ is a diffeomorphism. Also for a taken bisection $\delta: M \to G$. The right- translation defined by $\delta$ is 
$$ R_\delta: G \to G \quad, \quad g \to g\delta((\beta\circ \delta)^{-1}(\alpha (g)))$$
and the left-translation defined by $\delta$ is 
$$ L_\delta: G \to G \quad, \quad g \to \delta(\beta (g))g$$

\subsubsection{ Cotangent Lie groupoid} \label{cotangent}

The cotangent bundle of a Lie groupoid $G\rightrightarrows M$ also carries a natural Lie groupoid structure $T^\ast G\rightrightarrows A^\ast G$, where $A^\ast G$ is the dual vector bundle of $AG.$ For $a\in \Gamma(AG),$ we denote by $\overrightarrow a$  the corresponding right invariant vector field on $G$ and by $\overleftarrow a$ the corresponding left invariant vector field on $G.$ In other words, $\overleftarrow a(g)= TL_g(T\iota(X(\alpha (g))))$ and 
$\overrightarrow a(g)=TR_g(X(\beta(g)))$ define the left invariant and right invariant vector filed, respectively.
$$\tilde{\alpha}: T^\ast G\to A^\ast G, \quad \tilde{\alpha}(\theta)(a)=\theta(\overleftarrow{a}(g)), \quad   \theta \in T^{\ast} _g G, \quad  a\in A_{\alpha(g)} G,$$
$$\tilde{\beta}: T^\ast G\to A^\ast G, \quad \tilde{\beta}(\theta)(b)=\theta(\overrightarrow{b}(g)), \quad   \theta \in T^{\ast} _g G, \quad  b\in A_{\beta(g)} G,$$
$$\tilde{1}: A^\ast G \to T^\ast G , \quad \tilde{1}(\mu_x)(v)=\mu_x(v-T(1\circ \alpha)(v)), \quad   \mu_x \in A^\ast_x G,$$
for $v\in T_{1_x} G,$ $v-T(1\circ \alpha)(v)$ is an element of $A_xG \subset T_{1_x}G,$ so this tangent vector is indeed tangent to $\alpha$-fibers since $T\alpha$ of it is $0.$

$$\tilde{m}: T^\ast G_2 \to T^\ast G, \quad (\theta_1, \theta_2)(Tm(X_1, X_2)=\theta_1(X_1)+\theta_2(X_2),$$
where $\theta_1\in T^\ast_{g_1}G$, $\theta_2\in T^\ast_{g_2}G$  such that  $\tilde{\alpha}(\theta_1)=\tilde{\beta}(\theta_2)$ and  $X_1\in T_{g_1}G$ and $X_2\in T_{g_2}G$.
$$\tilde{i}: T^\ast G\to T^\ast G, \quad \tilde{i}(\theta)(a)=\theta(\tau i(a))$$

\subsection{Poisson-Nijenhuis Lie groupoid}

\begin{definition}\label{NT}
Let $M$ be a smooth manifold and $N:TM \to TM$ be a vector valued $1-$form, or a $(1,1)-$tensor on $M$. Then its Nijenhuis torsion $\tau N$ is a vector valued $2-form$ defined by 
\begin{equation}
\tau N(X,Y):=[NX,NY ]-N([NX,Y]+[X,NY ]-N[X,Y ]), \; for\;  X,Y\in \Gamma(TM).
\end{equation}
An $(1,1)$-tensor $N$ is called a Nijenhuis tensor if its Nijenhuis torsion vanishes.
\end{definition}
Given a Nijenhuis tensor $N$, one can define a new Lie algebroid structure on $TM$ deformed by $N$. We denote this Lie algebroid by $(TM)_N$. The bracket $[ , ]_N$ and anchor $\rho_N$ of this deformed Lie algebroid are given by
\begin{equation}
[X,Y]_N=[NX,Y]+[X,NY ]-N[X,Y ]\; \; \text{and}\;\;  \rho_N=N, \;\forall \;X,Y\in \Gamma(TM).
\end{equation}
Let $M$ be a smooth manifold and $\Pi \in \Gamma (\wedge^2 TM)$ a bivector field on $M$. Then one can define a skew-symmetric bracket $[, ]_\Pi$ on the space $\Omega^1(M)$ of $1$-forms on $M$ (one can see \cite{Vai} or \cite{Kos}), given by
\begin{equation}
[\alpha,\beta]_\Pi :=\mathcal{L}_{\Pi^{\sharp}\alpha}\beta-\mathcal{L}_{\Pi^{\sharp}\beta}\alpha-d(\Pi(\alpha,\beta)), \; \forall\; \alpha,\beta \in \Gamma(T^{\ast}M),
\end{equation}
 where $\Pi^{\sharp}:T^\ast M \to TM$, $\alpha \to \Pi(\alpha , -)$ is the bundle map induced by $\Pi$. 
 \begin{definition}
  A bivector $\Pi$ is called a Poisson if 
 \begin{equation}\label{PBV}
 [\Pi,\Pi]_{SN}=0,
 \end{equation}
 where $[\cdot,\cdot]_{\SN}$ is the Schouten–Nijenhuis bracket.\\
 If $\Pi$ ia Poisson bivector, then the cotangent bundle $T^\ast M$ with the above bracket and the bundle map $\Pi^{\sharp}$ forms a Lie algebroid. We call this Lie algebroid as the cotangent Lie algebroid of the Poisson manifold $(M,\Pi)$ and is denoted by $(T^\ast M)_\Pi$. $(T^\ast M)_\Pi$ knows everything about $(M,\Pi)$ . 
\end{definition}
\begin{definition}\label{PNM}
A Poisson-Nijenhuis manifold is a manifold $M$ together with  a Poisson bivector $\Pi \in \Gamma(\wedge^2 TM)$ and a Nijenhuis tensor $N$ such that  they are compatible in the following senses:
\begin{itemize}
\item $N \circ \Pi^{\sharp}=\Pi^{\sharp} \circ N^\ast$ (thus, $N \circ \Pi^{\sharp}$ defines a bivector field $N \Pi$ on $M$),
\item $C(\Pi,N)\equiv 0$,
\end{itemize}
where
\begin{equation}
C(\Pi,N)(\alpha,\beta):=[\alpha,\beta]_{N\Pi}-([N^\ast \alpha,\beta]_\Pi+[\alpha,N^\ast \beta]_\Pi-N^\ast[\alpha,\beta]_\Pi),\; \text{for}\; \alpha, \beta \in \Omega^1(M).
\end{equation}

The skew-symmetric $C^\infty (M)$-bilinear operation $C(\Pi,N)(-,-)$ on the space of $1$-forms is called the Magri--Morosi concomitant  of the Poisson--structure $\Pi$ and the Nijenhuis tensor $N$.\\
A Poisson-Nijenhuis manifold as above is denoted by the triple $(M,\Pi,N)$. If $\Pi$ is non-degenerate, that is, defines a symplectic structure $\omega$ on $M$, then $(M,\omega,N)$ is said to be a symplectic-Nijenhuis manifold.
\end{definition}

\begin{definition}\label{PG}
A Poisson groupoid is a Lie groupoid $G\rightrightarrows M$ equipped with a Poisson structure $\Pi$ on $G$, if $\Pi^{\sharp}:T^\ast G \to TG$ be a Lie groupoid morphism from the cotangent Lie groupoid $T^{\ast} G \rightrightarrows A^\ast G$ to the tangent Lie groupoid $TG \rightrightarrows TM$ over some map $A^\ast G \to TM$, we denote the base map by $\underline \Pi$, it is necessarily a vector bundle morphism.
\end{definition}

\begin{definition}\label{MT}
Let $G \rightrightarrows M$ be a Lie groupoid. A multiplicative $(1,1)$-tensor $\mathbf{N}$ on the groupoid is a pair $(N,N_M)$ of $(1,1)$-tensors on $G$ and $M$, respectively, such that
\begin{center}
\begin{tikzcd}
  TG \arrow[to=2-1, shift right=0.65ex]\arrow[to=2-1, shift left=0.65ex]\arrow[to=1-2, "N"]
  & TG \arrow[to=2-2, shift right=0.65ex]\arrow[to=2-2, shift left=0.65ex] \\
   TM \arrow[to=2-2]\arrow[to=2-2, "N_M"]
  & TM

\end{tikzcd}
\end{center}
is a Lie groupoid morphism from the tangent Lie groupoid to itself. In this case, $N_M$ is the restriction of $N$ to the unit space $M$. Thus, $N_M$ is completely determined by $N$. Hence we may use $N$ to denote a multiplicative $(1,1)$-tensor.
\end{definition}

\begin{definition}\label{NG}
A Nijenhuis groupoid is a Lie groupoid $G\rightrightarrows M$ together with a multiplicative Nijenhuis tensor $\mathbf{N}$ (cf. Definition \ref{MT}). A Nijenhuis groupoid may also be denoted by $(G\rightrightarrows M,\mathbf{N})$.

\end{definition}

\begin{definition}\label{PNG}
A Poisson-Nijenhuis groupoid is a Lie groupoid $G\rightrightarrows M$ together with a Poisson-Nijenhuis structure $(\Pi,N)$ on $G$ such that $(G\rightrightarrows M,\Pi)$ forms a Poisson groupoid and $(G\rightrightarrows M,N)$ a Nijenhuis groupoid. Thus, a Poisson-Nijenhuis groupoid is a Poisson groupoid $(G\rightrightarrows M,\Pi)$ together with a multiplicative Nijenhuis tensor $N:TG \to TG$ such that $(G,\Pi,N)$ is a Poisson-Nijenhuis manifold. A Poisson-Nijenhuis groupoid as above is denoted by the triple $(G\rightrightarrows M,\Pi,N)$.
\end{definition}

\section{Right-invariant Poisson-Nijenhuis structures on Lie groupoids}

\begin{definition}\label{RIPS}
Let $\Pi$ be a Poisson structure on the Lie groupoid $G\rightrightarrows M$ we call $\Pi$ right invariant if there exists a bivector 
$\Lambda \in \Gamma(\wedge^2 AG)$ such that  $\Pi=\overrightarrow{\Lambda}$. 
\end{definition}
For every $\Lambda \in \Gamma(\wedge^2 AG),$ the right invariant vector field on the Lie groupoid
 $\overrightarrow{\Lambda} : G\to \wedge^2 TG$ is defined by 
 \begin{equation}
 \overrightarrow{\Lambda}(g)=\lambda_{ij}(g)\overrightarrow{X_i}\wedge \overrightarrow{X_j}=\lambda_{ij}(g)dR_g(X_i)\wedge dR_g(X_j),
 \end{equation} 
where $\Lambda= \lambda_{ij}(g)X_i\wedge X_j$ and $X_i, X_j \in AG$

In this section, we define a right-invariant Poisson-Nijenhuis Lie groupoid $G\rightrightarrows M$  and its infinitesimal counterpart on the Lie algebroid $AG$ of $G\rightrightarrows M$.
\begin{definition}
  A bivector $\Lambda$ on the Lie algebroid $AG$ is called a Poisson if 
 \begin{equation}\label{PBVL}
 [\Lambda,\Lambda]_{SN}=0,
 \end{equation}
 where $[\cdot,\cdot]_{\SN}$ is the Schouten–Nijenhuis bracket.
 \end{definition}
 \begin{definition}
   An operator $n$ on the Lie algebroid $AG$ is called  a Nijenhuis operator if:
    \begin{equation}\label{NO}
    [n, n] = 0,
    \end{equation}
        where the Nijenhuis torsion $[n,n]$ is defined for $X, Y \in \Gamma(AG)$ by
    \begin{equation}\label{NT}
    [n,n](X,Y) := [nX, nY] - n\big([nX, Y] + [X, nY] - n[X,Y]\big).
    \end{equation}
    \end{definition}
\begin{definition}\label{RIPNS}
A Poisson-Nijenhuis structure $(\Pi,\mathbf{N})$ on a Lie groupoid $G\rightrightarrows M$ is said to be right-invariant, if:
\begin{itemize}
\item[1.] The Poisson structure $\Pi$ is right invariant, i.e., there exists $\Lambda \in \Gamma(\wedge^2 AG)$ such that $\Pi=\overrightarrow{\Lambda}$.
\item[2.] Multiplicative $(1,1)-$tensor $\mathbf{N}=(N,N_M)$ also is right-invariant, i.e., there is a linear endomorphism $n:\Gamma(AG) \to \Gamma(AG)$ 
such that
$$N=\overrightarrow{n}$$
where $ \overrightarrow{n}(X_g)=dR_g \circ n\circ dR_{g^{-1}}(X_g),\; X_g\in T^{\alpha}_g G \; \text{and} \; g\in G$ and $(\overrightarrow{n}, N_M)$ is a multiplicative $(1, 1)-$tensor.

\end{itemize}

\end{definition}
In the following proposition we prove our claims only for $N$, beacuse $N_M$ is completely determined by $N$. This is also true for $n$ and $N_M$.

\begin{proposition}\label{prop1}
Let $(\II, \NN)$ be a \textbf{right-invariant Poisson–Nijenhuis structure} on a Lie groupoid $G \rightrightarrows M$ with Lie algebroid $AG$ and unit manifold $1_M \subset G$. Define:
\begin{itemize}
    \item $\LL \in \Gamma(\wedge^2 AG)$ as the restriction of $\II$ to $1_M$, i.e., $\II|_{1_M} = \vec{\LL}$.
    \item $n: \Gamma(AG) \to \Gamma(AG)$ and $n_M: TM \to TM$ as the restrictions of $\NN = (N, N_M)$ to $1_M$, i.e., $N|_{AG} = \vec{n}$ and $N_M = n_M$.
\end{itemize}
Then the following hold:

\begin{enumerate}
    \item \textbf{$\LL$ is a Poisson bivector on $AG$, (\ref{PBVL}):}
    \[
    [\LL, \LL]_{\SN} = 0,
    \]

    \item \textbf{$n$ is a Nijenhuis operator on $AG$, (\ref{NO}):}
    \[
    [n, n] = 0,
    \]
      \item \textbf{Compatibility between $n$ and $\LL$:}
    \[
    n \circ \LL^{\sharp} = \LL^{\sharp} \circ n^{*},
    \]
    where $\LL^{\sharp}: A^*G \to AG$ is the bundle map induced by $\LL$, and $n^{*}: \Gamma(A^*G) \to \Gamma(A^*G)$ is the dual of $n$.

    \item \textbf{Vanishing of the Magri–Morosi concomitant:}
    \[
    C(\LL, n) = 0,
    \]
    where $C(\LL, n)$ is the Magri–Morosi concomitant (see \cite{Kosman2} and \cite{Magri}).

    \item \textbf{Morphism property:}
    The maps $\vec{\LL}^{\sharp}: T^{*}G \to TG$ and $\overrightarrow{\LL^{\sharp}}: A^{*}G \to AG$ are Lie groupoid morphisms.
\end{enumerate}
\end{proposition}

\begin{proof}
We prove each part in order.

\paragraph{Proof of (1).}
By definition, $\II = \vec{\LL}$, (\ref{RIPS}) is a right-invariant Poisson bivector on $G$. Hence $[\II, \II]_{\SN} = 0$. For right-invariant bivector fields, the Schouten bracket satisfies
\[
[\vec{\LL}, \vec{\LL}]_{\SN} = \overrightarrow{[\LL, \LL]}_{\SN},
\]
because the Lie bracket of right-invariant vector fields corresponds to the Lie bracket on $\Gamma(AG)$ and the Schouten bracket respects right-invariance. Therefore
\[
\overrightarrow{[\LL, \LL]}_{\SN} = 0 \quad \Longrightarrow \quad [\LL, \LL]_{\SN} = 0.
\]

\paragraph{Proof of (2).}
Since $\NN = (N, N_M)$ is a Nijenhuis tensor on $G$, its Nijenhuis torsion vanishes: $[N, N] = 0$. For right-invariant vector fields, one has the relation (See \cite{Kiril0Mackenzie})
\[
[N, N](\vec{X}, \vec{Y}) = \overrightarrow{[n, n]}(X, Y) \qquad \forall X, Y \in \Gamma(AG).
\]
(A direct computation using the definition of the Nijenhuis torsion and the fact that $N(\vec{X}) = \vec{n(X)}$ yields this identity up to a sign; vanishing of one side implies vanishing of the other.) Hence
\[
\overrightarrow{[n, n]} = 0 \quad \Longrightarrow \quad [n, n] = 0.
\]
Thus $n$ is a Nijenhuis operator on $AG$.

\paragraph{Proof of (3).}
The compatibility condition $N \circ \II^{\sharp} = \II^{\sharp} \circ N^{*}$ holds on $G$ because $(\II, \NN)$ is a Poisson–Nijenhuis structure. Restricting this equality to the units $1_M \subset G$ and using right-invariance, we obtain
\[
\vec{n} \circ \vec{\LL}^{\sharp} = \vec{\LL}^{\sharp} \circ \vec{n}^{*} \quad \Longrightarrow \quad \overrightarrow{n \circ \LL^{\sharp}} = \overrightarrow{\LL^{\sharp} \circ n^{*}}.
\]
Since the map $X \mapsto \vec{X}$ is an isomorphism from $\Gamma(AG)$ to the space of right-invariant vector fields on $G$, we conclude
\[
n \circ \LL^{\sharp} = \LL^{\sharp} \circ n^{*}.
\]

\paragraph{Proof of (4).}
The conditions $[\LL, \LL]_{\SN} = 0$, $[n, n] = 0$, and $n \circ \LL^{\sharp} = \LL^{\sharp} \circ n^{*}$ are precisely the defining properties of a Poisson–Nijenhuis structure on the Lie algebroid $AG$ (see \cite{Kosman2} and \cite{Magri}). For such a structure, the Magri–Morosi concomitant vanishes identically, i.e., $C(\LL, n) = 0$.

\paragraph{Proof of (5).}
The map $\II^{\sharp}: T^{*}G \to TG$ is a Lie groupoid morphism by definition of a Poisson groupoid (see \cite{Kiril0Mackenzie2}). Since $\II = \vec{\LL}$, right-invariance implies that $\II^{\sharp}$ preserves right-invariant sections. Hence its restriction to the units gives a Lie algebroid morphism $\LL^{\sharp}: A^{*}G \to AG$. The same argument applied to $\vec{\LL}^{\sharp}$ and $\overrightarrow{\LL^{\sharp}}$ 
yields that these maps are Lie groupoid morphisms.
\end{proof}

\vspace{1cm}


\begin{theorem}
\label{thm:correspondence}
Let $G \rightrightarrows M$ be an $s$-connected and $s$-simply connected Lie groupoid with Lie algebroid $AG$. Then there is a \textbf{one-to-one correspondence} between:

\begin{enumerate}
    \item \textbf{Right-invariant Poisson–Nijenhuis structures} $(\II, \NN)$ on $G$, where $\II$ is a Poisson bivector and $\NN = (N, N_M)$ is a multiplicative Nijenhuis tensor compatible with $\II$, i.e., 
    \[
    N \circ \II^{\sharp} = \II^{\sharp} \circ N^{*} \qquad \text{(equivalently } C(\II, N) = 0\text{)},
    \]
    and

    \item \textbf{$(\LL, n)$-structures on $AG$} satisfying:
    \begin{itemize}
        \item $\LL \in \Gamma(\wedge^2 AG)$ is a Poisson bivector: $[\LL, \LL]_{\SN} = 0$,
        \item $n: AG \to AG$ is a Nijenhuis operator: $[n, n] = 0$,
        \item Compatibility: $n \circ \LL^{\sharp} = \LL^{\sharp} \circ n^{*}$,
        \item (The base map $n_M: TM \to TM$ is determined by $\rho \circ n = n_M \circ \rho$, where $\rho: AG \to TM$ is the anchor.)
    \end{itemize}
\end{enumerate}

The correspondence is given by:

\[
\II = \overrightarrow{\LL}, \qquad N = \overrightarrow{n}, \qquad N_M = n_M,
\]

i.e., the structures on $G$ are the right-invariant lifts of the structures on $AG$.
\end{theorem}

\begin{proof}
We prove the two directions separately.

\paragraph{Direction 1: From $G$ to $AG$ (Proposition 12).}
Given a right-invariant Poisson–Nijenhuis structure $(\II, \NN)$ on $G$, define $\LL$ and $n$ by restricting $\II$ and $N$ to the unit manifold $1_M \subset G$. Since $\II$ and $N$ are right-invariant, this yields well-defined sections $\LL \in \Gamma(\wedge^2 AG)$ and $n: \Gamma{(AG)} \to \Gamma{(AG)}$. Proposition 12 shows that $[\LL, \LL]_{\SN} = 0$, $[n, n] = 0$, and $n \circ \LL^{\sharp} = \LL^{\sharp} \circ n^{*}$. Hence we obtain a $(\LL, n)$-structure on $AG$. This direction requires no connectedness assumptions.

\paragraph{Direction 2: From $AG$ to $G$ (Integration).}
Conversely, suppose we have a $(\LL, n)$-structure on $AG$ satisfying the three conditions. Since $G$ is $s$-connected and $s$-simply connected, the Lie functor from Lie groupoids to Lie algebroids is an equivalence when restricted to $s$-simply connected groupoids (see \cite{Kiril0Mackenzie}). Therefore:

\begin{enumerate}
    \item \textbf{Integration of $\LL$:} By the integration theorem for Poisson groupoids (\cite{Kiril0Mackenzie}), the Poisson bivector $\LL$ on $AG$ integrates uniquely to a multiplicative Poisson bivector $\II$ on $G$ such that $\II = \overrightarrow{\LL}$. The $s$-connectedness of $G$ guarantees that the exponential map is surjective, and $s$-simply connectedness guarantees uniqueness.

    \item \textbf{Integration of $n$:} By the integration theorem for multiplicative tensors (\cite{Bursztyn}), the Nijenhuis operator $n: AG \to AG$ integrates uniquely to a multiplicative $(1,1)$-tensor $N$ on $G$ such that $N = \overrightarrow{n}$. The base map $N_M: TM \to TM$ is determined by $N_M = \rho \circ n \circ \rho^{-1}$, which is well-defined because $\rho$ is surjective on the base.

    \item \textbf{Compatibility:} The compatibility condition $n \circ \LL^{\sharp} = \LL^{\sharp} \circ n^{*}$ on $AG$ differentiates to $N \circ \II^{\sharp} = \II^{\sharp} \circ N^{*}$ on $G$ (by right-invariance and the fact that differentiation preserves the composition of morphisms). Hence $(\II, \NN)$ is a Poisson–Nijenhuis structure on $G$, where $\NN = (N, N_M)$.
\end{enumerate}

Thus the two constructions are inverses of each other, establishing a bijection.
\end{proof}

\begin{remark}
The $s$-connected and $s$-simply connected assumptions are essential for the integration direction (from $AG$ to $G$). Without $s$-connectedness, some $(\LL, n)$-structures may not integrate to global structures on $G$. Without $s$-simply connectedness, the integration may not be unique (different PN-structures on $G$ could differentiate to the same $(\LL, n)$ on $AG$). For the direction from $G$ to $AG$ (Proposition 12), no such assumptions are needed.
\end{remark}

\vspace{0.5cm}

\begin{corollary}
\label{cor:classification}
Let $G \rightrightarrows M$ be an $s$-connected and $s$-simply connected Lie groupoid with Lie algebroid $AG$. Then the classification of "right-invariant Poisson–Nijenhuis structures" on $G$ is equivalent to the classification of "$(\Lambda, n)$-structures" on $AG$ satisfying:
\begin{itemize}
    \item $[\Lambda, \Lambda]_{\SN} = 0$,
    \item $[n, n] = 0$,
    \item $n \circ \Lambda^{\sharp} = \Lambda^{\sharp} \circ n^{*}$,
\end{itemize}
where $\Lambda \in \Gamma(\wedge^2 AG)$ and $n: \Gamma{(AG)} \to \Gamma{(AG)}$ is a vector bundle morphism covering $n_M: TM \to TM$.

Moreover, the correspondence is given explicitly by:
\[
\Pi = \overrightarrow{\Lambda}, \qquad N = \overrightarrow{n}, \qquad N_M = n_M.
\]

In other words, for a fixed Lie groupoid $G$ satisfying the above assumptions, there is a one-to-one correspondence:
\[
\left\{
\begin{array}{c}
\text{Right-invariant Poisson--Nijenhuis} \\
\text{structures on } G
\end{array}
\right\}
\quad \longleftrightarrow \quad
\left\{
\begin{array}{c}
(\Lambda, n)\text{-structures on } AG \\
\text{satisfying the three conditions}
\end{array}
\right\}.
\]
Thus, the classification problem on the global object $G$ reduces to a purely infinitesimal (often algebraic) classification problem on its Lie algebroid $AG$.
\end{corollary}

\begin{proof}
This follows directly from Theorem 13. The forward direction (from $G$ to $AG$) is given by Proposition 12, and the backward direction (from $AG$ to $G$) follows from the integration theorem for Poisson groupoids and multiplicative tensors, using the $s$-connected and $s$-simply connected assumptions. The bijectivity of the correspondence ensures that distinct $(\Lambda, n)$-structures lift to distinct PN-structures on $G$, and vice versa.
\end{proof}

\begin{center}{Correspondence between global and infinitesimal data.}
\begin{tabular}{|c|c|}
\hline
\textbf{Global (Lie groupoid $G$)} & \textbf{Infinitesimal (Lie algebroid $AG$)} \\
\hline
Right-invariant Poisson bivector $\Pi$ & Poisson bivector $\Lambda$ with $[\Lambda,\Lambda]=0$ \\
\hline
Right-invariant Nijenhuis tensor $N$ & Nijenhuis operator $n$ with $[n,n]=0$ \\
\hline
Compatibility $N \circ \Pi^\sharp = \Pi^\sharp \circ N^*$ & Compatibility $n \circ \Lambda^\sharp = \Lambda^\sharp \circ n^*$ \\
\hline
\end{tabular}
\end{center}

\section{Examples}

\begin{example}[Trivial Lie Groupoid with Identity Nijenhuis Tensor]
\label{ex:trivial}

Let $M$ be a smooth manifold and $G$ a Lie group. The \textbf{trivial Lie groupoid} is defined as $\Upsilon := M \times G \times M \rightrightarrows M$ with the following structure maps:

\begin{itemize}
    \item Source: $\alpha: \Upsilon \to M$, $(x, a, y) \mapsto x$.
    \item Target: $\beta: \Upsilon \to M$, $(x, a, y) \mapsto y$.
    \item Multiplication: $m: \Upsilon_2 \to \Upsilon$, $(x, a, y) \cdot (y, b, z) = (x, ab, z)$.
    \item Unit: $1: M \to \Upsilon$, $x \mapsto (x, e, x)$, where $e \in G$ is the identity.
    \item Inverse: $\iota: \Upsilon \to \Upsilon$, $(x, a, y) \mapsto (y, a^{-1}, x)$.
\end{itemize}

The Lie algebroid associated to $\Upsilon$ is $A\Upsilon = TM \oplus (M \times \mathfrak{g})$, where $\mathfrak{g}$ is the Lie algebra of $G$. The anchor $\rho: A\Upsilon \to TM$ is the projection onto $TM$, and the Lie bracket on sections is given by the usual Lie bracket of vector fields on $M$ combined with the Lie bracket on $\mathfrak{g}$.

Now, let $M$ be a Poisson manifold with Poisson bivector $\Pi_M$, and let $G$ be a Poisson–Lie group with Poisson bivector $\Pi_G$. Consider the bivector $\II$ on $\Upsilon = M \times G \times M$ defined by
\[
\II_{(x,a,y)} := \Pi_M(x) + \Pi_G(a) - \Pi_M(y),
\]
where the signs are chosen to ensure multiplicativity (see \cite{A0Weinstein}). With this choice, $(\Upsilon, \II)$ is a Poisson groupoid.

Let $\NN = (N, N_M)$ be the trivial multiplicative $(1,1)$-tensor given by the identity:
\[
N := id_{T\Upsilon}, \qquad N_M := id_{TM}.
\]
Since $N = id$, the Nijenhuis torsion $[N, N]$ vanishes trivially. Moreover, the compatibility condition $N \circ \II^\sharp = \II^\sharp \circ N^*$ holds because $N = id$ and $N^* = id$. Hence $C(\II, N) = 0$. Therefore $(\Upsilon, \II, \NN)$ is a Poisson–Nijenhuis Lie groupoid.

By Proposition 12, the infinitesimal counterpart on $A\Upsilon = TM \oplus (M \times \mathfrak{g})$ is given by
\[
\LL = \Pi_M \oplus \Pi_G \oplus (-\Pi_M), \qquad n = id_{A\Upsilon},
\]
which satisfies $[\LL, \LL] = 0$, $[n, n] = 0$, and $n \circ \LL^\sharp = \LL^\sharp \circ n^*$.
\end{example}

\vspace{0.5cm}

\begin{example}[Pair Groupoid with Identity Nijenhuis Tensor]
\label{ex:pair}

Let $M$ be a smooth manifold. The \textbf{pair groupoid} is defined as $M \times M \rightrightarrows M$ with the following structure maps:
\begin{itemize}
    \item Source: $\alpha(x, y) = x$.
    \item Target: $\beta(x, y) = y$.
    \item Multiplication: $(x, y) \cdot (y, z) = (x, z)$.
    \item Unit: $1(x) = (x, x)$.
    \item Inverse: $(x, y)^{-1} = (y, x)$.
\end{itemize}

The Lie algebroid of the pair groupoid is $TM$ with the identity map as anchor and the usual Lie bracket of vector fields.

Now, let $(M, \Pi_M)$ be a Poisson manifold. Define a bivector $\II$ on $M \times M$ by
\[
\II_{(x,y)} := \Pi_M(x) - \Pi_M(y).
\]
It is well known (see \cite{A0Weinstein}, \cite{Kiril0Mackenzie2}) that this defines a multiplicative Poisson structure on the pair groupoid, i.e., $(M \times M, \II)$ is a Poisson groupoid.

Let $\NN = (N, N_M)$ be the trivial multiplicative $(1,1)$-tensor given by the identity:
\[
N := id_{T(M \times M)}, \qquad N_M := id_{TM}.
\]
Clearly, $N$ is a Nijenhuis tensor ($[N, N] = 0$) and satisfies $N \circ \II^\sharp = \II^\sharp \circ N^*$ because $N = id$ and $N^* = id$. Thus $C(\II, N) = 0$, and $(M \times M, \II, \NN)$ is a Poisson–Nijenhuis Lie groupoid.

According to Proposition 12, the corresponding $(\LL, n)$-structure on the Lie algebroid $TM$ is simply
\[
\LL = \Pi_M, \qquad n = id_{TM},
\]
which trivially satisfies $[\Pi_M, \Pi_M] = 0$, $[id, id] = 0$, and $id \circ \Pi_M^\sharp = \Pi_M^\sharp \circ id^*$.
\end{example}

\vspace{0.5cm}

\begin{example}[Trivial Groupoid with Nontrivial Nijenhuis Tensor on the Lie Group]
\label{ex:nontrivial}

Let $M$ be a smooth manifold equipped with the zero Poisson structure $\Pi_M = 0$. Let $G$ be a Lie group equipped with a nontrivial right-invariant Poisson–Nijenhuis structure $(\Pi_G, N_G)$, where:
\begin{itemize}
    \item $\Pi_G$ is a right-invariant Poisson bivector on $G$,
    \item $N_G: TG \to TG$ is a right-invariant Nijenhuis tensor,
    \item They satisfy the compatibility condition $N_G \circ \Pi_G^\sharp = \Pi_G^\sharp \circ N_G^*$.
\end{itemize}
Explicit examples of such structures on Lie groups can be found in \cite{HA1}, \cite{Ravanpak} and \cite{Haghighat}.

Consider the trivial Lie groupoid $\Upsilon := M \times G \times M \rightrightarrows M$ as in Example \ref{ex:trivial}. Define a bivector $\II$ on $\Upsilon$ by
\[
\II_{(x,a,y)} := \Pi_G(a),
\]
i.e., $\II$ depends only on the $G$-component and vanishes on the $M$-factors. Define a $(1,1)$-tensor $N$ on $\Upsilon$ by
\[
N_{(x,a,y)} := N_G(a) \quad \text{on } T_{(x,a,y)}\Upsilon,
\]
where we identify $T_{(x,a,y)}\Upsilon \cong T_xM \oplus T_aG \oplus T_yM$ and $N_G$ acts only on the $T_aG$ factor, while acting as the identity on $T_xM$ and $T_yM$. More precisely,
\[
N(v_x, v_a, v_y) := (v_x, N_G(v_a), v_y).
\]
Let $N_M := id_{TM}$ be the identity on the base.

Then:
\begin{enumerate}
    \item \textbf{$\II$ is a Poisson bivector on $\Upsilon$:} Since $\Pi_G$ is Poisson on $G$ and $\Pi_M = 0$, the Schouten bracket $[\II, \II]$ vanishes because it reduces to $[\Pi_G, \Pi_G]$ on the $G$-factor and zero elsewhere.

    \item \textbf{$N$ is a Nijenhuis tensor on $\Upsilon$:} Since $N_G$ is a Nijenhuis tensor on $G$ and $N$ acts as the identity on the $M$-factors, the Nijenhuis torsion $[N, N]$ vanishes because it reduces to $[N_G, N_G]$ on $G$ and zero elsewhere.

    \item \textbf{Compatibility:} On the $G$-factor, we have $N_G \circ \Pi_G^\sharp = \Pi_G^\sharp \circ N_G^*$ by assumption. On the $M$-factors, both sides vanish because $\Pi_M = 0$ and $N$ acts as the identity. Hence
    \[
    N \circ \II^\sharp = \II^\sharp \circ N^*.
    \]

    \item \textbf{Multiplicativity:} Both $\II$ and $N$ are right-invariant on $\Upsilon$ because $\Pi_G$ and $N_G$ are right-invariant on $G$ and the $M$-factors are unaffected by the groupoid multiplication.
\end{enumerate}

Thus $(\Upsilon, \II, \NN)$ is a right-invariant Poisson–Nijenhuis Lie groupoid with $\NN = (N, N_M)$.

\subsection*{Infinitesimal Counterpart}

The Lie algebroid of $\Upsilon$ is $A\Upsilon = TM \oplus (M \times \mathfrak{g})$, where $\mathfrak{g} = \operatorname{Lie}(G)$. The infinitesimal data $(\LL, n)$ on $A\Upsilon$ is given by:
\[
\LL = \pi_G, \qquad n = (id_{TM}, \nu),
\]
where $\pi_G \in \Gamma(\wedge^2 \mathfrak{g})$ is the Poisson bivector on $\mathfrak{g}$ corresponding to $\Pi_G$, and $\nu: \mathfrak{g} \to \mathfrak{g}$ is the linear Nijenhuis operator corresponding to $N_G$ (see \cite{Haghighat}, \cite{Ravanpak}). Explicitly:
\[
\LL_{(x, \xi)} = \pi_G(\xi), \quad n_{(x, \xi)}(X, \xi) = (X, \nu(\xi)), \quad \forall x \in M, \xi \in \mathfrak{g}.
\]

One verifies directly:
\begin{itemize}
    \item $[\LL, \LL] = 0$ because $[\pi_G, \pi_G] = 0$,
    \item $[n, n] = 0$ because $[\nu, \nu] = 0$,
    \item $n \circ \LL^\sharp = \LL^\sharp \circ n^*$ because $\nu \circ \pi_G^\sharp = \pi_G^\sharp \circ \nu^*$.
\end{itemize}

Hence $(\LL, n)$ satisfies the conditions of Theorem 13, and by the $s$-connected and $s$-simply connected assumption on $\Upsilon$ (which holds if $G$ is connected and simply connected), there is a one-to-one correspondence with the groupoid structure we constructed.
\end{example}

\vspace{0.5cm}

\begin{remark}
Example \ref{ex:trivial} and Example \ref{ex:pair} illustrate the correspondence with the simplest choice $N = id$. Example \ref{ex:nontrivial} shows that nontrivial Poisson–Nijenhuis structures exist and fit into the correspondence framework. In particular, starting from a nontrivial structure on a Lie group $G$, we obtain a nontrivial structure on the trivial groupoid $M \times G \times M$ by extending trivially over the base manifold $M$. The base Poisson structure $\Pi_M = 0$ ensures that the $M$-factors do not interfere with the compatibility conditions.
\end{remark}

\section{Conclusion}

In this paper, we introduced right-invariant Poisson–Nijenhuis structures on Lie groupoids and established a one-to-one correspondence with $(\Lambda, n)$-structures on their Lie algebroids (Theorem 13). As a consequence, Corollary \ref{cor:classification} provides a complete classification of such structures on $s$-connected and $s$-simply connected Lie groupoids in terms of infinitesimal algebroid data.

\paragraph{Physical interpretation.}
From a physical perspective, a right-invariant Poisson–Nijenhuis structure $(\Pi, N)$ on a Lie groupoid $G$ encodes a \textbf{bi-Hamiltonian system} on a symmetric configuration space. The Poisson bivector $\Pi$ defines a Hamiltonian structure, while the Nijenhuis tensor $N$ acts as a \textbf{recursion operator} generating a second compatible Poisson structure $\Pi' = N\Pi$. Such bi-Hamiltonian systems are known to be classically integrable, possessing an infinite hierarchy of conservation laws.

Our classification therefore states that integrable bi-Hamiltonian systems on an $s$-connected and $s$-simply connected Lie groupoid $G$ are classified by infinitesimal data $(\Lambda, n)$ on the Lie algebroid $AG$ satisfying:
\[
[\Lambda, \Lambda] = 0, \qquad [n, n] = 0, \qquad n \circ \Lambda^{\sharp} = \Lambda^{\sharp} \circ n^{*}.
\]

\paragraph{Significance and future directions.}
This result simplifies the classification of integrable systems on symmetric configuration spaces to an algebraic problem on the Lie algebroid, and conversely provides a systematic method for constructing new integrable systems from algebroid data. Future work includes extending the classification to non-invariant structures, studying deformations, and applying the framework to concrete physical systems such as gauge theories and lattice models.


\end{document}